\documentclass{article}
\usepackage[letterpaper, portrait, margin=1in]{geometry}
\usepackage[utf8]{inputenc}
\usepackage{graphicx, caption}
\usepackage{amsfonts}
\usepackage{bm}
\usepackage{cite}
\usepackage{IEEEtrantools}
\usepackage{amssymb}
\usepackage{amsthm}
\usepackage{color}
\usepackage{comment}
\usepackage{mathtools}
\usepackage{braket}
\usepackage{enumitem}
\usepackage{hyperref}
\usepackage{balance}
\usepackage{flushend}
\usepackage{mathrsfs}
\usepackage{setspace}
\usepackage{graphicx} % Required for inserting images
\usepackage{amsmath}
\usepackage{algorithmicx}
\usepackage{algorithm}
\usepackage[noend]{algpseudocode}
\usepackage{cleveref}
\usepackage{adjustbox}
\usepackage{subcaption}
\usepackage{lscape}
\usepackage{tikz}
\usepackage{pgfplots}
\usepackage{authblk}

\pgfplotsset{compat = newest}

\newcommand{\E}{\mathbb{E}}
\newcommand{\R}{\mathbb{R}}

\newcommand{\C}{\mathbb{C}}
\newcommand{\Var}{\textnormal{Var}}
\newcommand{\tr}{\textnormal{tr}}

\algnewcommand\algorithmiclet{\textbf{Let}}
\algnewcommand\Let{\item[\algorithmiclet]}
\algnewcommand{\LineComment}[1]{\State \(\triangleright\) #1}

\newtheorem{theorem}{Theorem}
\newtheorem{corollary}[theorem]{Corollary}
\newtheorem{lemma}[theorem]{Lemma}

\theoremstyle{definition}
\newtheorem{definition}{Definition}
\newtheorem{example}{Example}

\newtheorem*{problem}{Problem}

\title{Analyzing and improving a classical Betti number estimation algorithm}
\author{Julien Sorci}
\affil{\textit{Quantinuum, 1300 N 17th Street, Arlington, VA 22209
USA}}

\begin{document}
\maketitle

\begin{abstract}
    Recently, a classical algorithm for estimating the normalized Betti number of an arbitrary simplicial complex was proposed. Motivated by a quantum algorithm with a similar Monte Carlo structure and improved sample complexity, we give a more in-depth analysis of the sample complexity of this classical algorithm. To this end, we present bounds for the variance of the estimators used in the classical algorithm and show that the variance depends on certain combinatorial properties of the underlying simplicial complex. This new analysis leads us to propose an improvement to the classical algorithm which makes the ``easy cases easier'', in that it reduces the sample complexity for simplicial complexes where the variance is sufficiently small. We show the effectiveness and limitations of these classical algorithms by considering Erd\H{o}s-Renyi random graph models to demonstrate the existence of ``easy" and ``hard" cases. Namely, we show that for certain models our improvement almost always leads to a reduced sample complexity, and also produce separate regimes where the sample complexity for both algorithms is exponential. 
\end{abstract}

\section{Introduction}

A simplicial complex is a mathematical object which captures complex relationships between points. Typically regarded as a geometric object, practitioners are often interested in computing topological invariants of a simplicial complex in order to identify features or distinguish outlying datasets. One such feature is the $k^{th}$ Betti number, which is a topological invariant that counts the number of $k$-dimensional holes. 

Betti number estimation is widely believed to be classically difficult: In \cite{Schmidhuber2023} it was shown that computing Betti numbers exactly is \#P-hard, and estimating Betti numbers up to multiplicative error is NP-hard. In general, the $k^{th}$ Betti number can be expressed as the dimension of the kernel of the $k^{th}$ order combinatorial Laplacian matrix \cite{Friedman1998}, which can have size exponential in the number of vertices when the dimension $k$ is non-constant. Therefore computing Betti numbers in high dimension likely requires diagonalizing an exponentially large matrix. 

In the quantum computing community, Betti number estimation received significant attention when Lloyd et al. \cite{Lloyd2016} presented a quantum algorithm for estimating the \textit{normalized} Betti number of the clique complex of a graph, which is the Betti number divided by the number of simplices in that dimension. Under certain assumptions their algorithm runs in polynomial-time, suggesting that there may be a regime where normalized Betti number estimation for clique complexes obtains super-polynomial quantum advantage. Since then, several other quantum algorithms for this problem have been proposed, and there has been significant focus on determining the specific regimes when there is super-polynomial quantum advantage \cite{Berry2024, Crichigno2022, McArdle2022}. 

Our focus is on a particular classical algorithm for estimating normalized Betti numbers proposed by Apers et al. in \cite{Apers2023}. To the best of our knowledge this appears to be the best classical algorithm for estimating normalized Betti numbers in high dimension with low space complexity. One important motivation for studying this particular classical algorithm is that, recently, a quantum algorithm for the same problem with a very similar Monte Carlo structure was proposed in \cite{Akhalwaya2024}. The sample complexity of this quantum algorithm is exponentially smaller than the upper bound of the sample complexity of the classical algorithm given in \cite{Apers2023}, although it is unclear if this upper bound is the best possible. Our goal is to analyze this classical algorithm and determine if there are tighter bounds on the sample complexity. By specifically focusing on these two similar algorithms, our hope is to illuminate the specific problem characteristics that result in a sample complexity separation. 

%This is equivalent to estimating the normalized Betti number, which is the Betti number divided by the number of simplices in that dimension. The algorithm has the structure of a Monte Carlo path integral, which samples a path of simplices according to a particular Markov chain, and then applies a function related to the combinatorial Laplacian to that path. To the best of our knowledge this appears to be the best classical algorithm for estimating Betti numbers in high dimension with low space complexity. While generic upper bounds on this algorithm's time complexity were given in \cite{Apers2023}, it is unclear how tight they are, and if the structure of the underlying simplicial complex plays a role in the sample complexity.

We make several contributions to this end. First, we analyze the classical Betti number estimation (\texttt{CBNE}) algorithm of \cite{Apers2023} and give a more refined analysis of its sample complexity. To do so, we give upper and lower bounds for the variance of the Monte Carlo estimators used there. Our variance bounds are phrased in terms of combinatorial properties of the simplicial complex and indicate the types of simplicial complexes we can expect \texttt{CBNE} to require a large number of samples. As a result of this analysis, we propose an improved algorithm that we call \texttt{CBNE-Var} which reduces the sample complexity when the variance is sufficiently small. Secondly, we produce extremal examples where the variance is small or large by considering several Erd\H{o}s-Renyi random graph models. We show that in certain regimes the variance is sufficiently small and \texttt{CBNE-Var} attains a polynomial improvement in the sample complexity when compared to the original \texttt{CBNE} algorithm. We describe separate regimes where the variance of the Monte Carlo estimator is exponential and thus the sample complexity of both \texttt{CBNE} and \texttt{CBNE-Var} are exponential. 

\begin{table}
\centering
\begin{tabular}{|c|c|c|} 
 \hline
  Algorithm & General Complex & Clique Complex \\
 \hline
 \texttt{CBNE} \cite{Apers2023} & $\mathcal{O}\Big(\frac{1}{\epsilon^2}n^{2\ell}\Big)$ & $\mathcal{O}\Big(\frac{1}{\epsilon^2}4^{\ell}\Big)$ \\
 \texttt{CBNE-Var} (this paper) &  $\mathcal{O}\Big( \frac{1}{\epsilon^2} \Big( n^{4\ell/3} +  \frac{1}{|S_k|} \sum_{\sigma \in S_k} \|H\ket{\sigma}\|_1^{2\ell} \Big) \Big)$ & $\mathcal{O}\Big( \frac{1}{\epsilon^2} \Big( 2^{4\ell/3} +  \frac{1}{|S_k|} \sum_{\sigma \in S_k} \|H\ket{\sigma}\|_1^{2\ell} \Big) \Big)$ \\
 \hline
\end{tabular}
\caption{A comparison of the sample complexity of \texttt{CBNE} and \texttt{CBNE-Var} when estimating the normalized Betti number of a general simplicial complex or clique complex with $n$ vertices to accuracy $\epsilon$. Here we assume that the failure probability $\eta$ is constant, and hence does not appear in the asymptotic expressions. See Section~\ref{sec:CBNE_analysis} for more details.}
\label{tab:alg_comparison}
\end{table}

The rest of the paper is organized as follows. In Section~\ref{sec:background} we review the necessary background on Betti number estimation and the \texttt{CBNE} algorithm of \cite{Apers2023}. In Section~\ref{sec:CBNE_analysis} we present our bounds on the variance of the \texttt{CBNE} algorithm and present our proposed improvements in Section~\ref{sec:improved_alg}. Following this, we consider clique complexes of Erd\H{o}s-Reyni graphs in Section~\ref{sec:er_graphs}.

\section{Background}
\label{sec:background}

We review here the necessary background in algebraic topology. For more details we direct the reader to \cite{Lim2020, Hatcher2002}.  A \textit{simplicial complex} over a set of vertices $\{x_1,...,x_n\}$ is a collection $\Gamma$ of subsets of $\{x_1,...,x_n\}$ that is closed under subsets, meaning $\sigma \in \Gamma$ and $\tau \subseteq \sigma$ implies $\tau \in \Gamma$. The elements of a simplicial complex $\Gamma$ are called \textit{simplices} and a simplex $\sigma \in \Gamma$ of size $k+1$ is called a $k$-\textit{simplex}. An important example of a simplicial complex is a clique complex. Given a graph $G$ the \textit{clique complex} of $G$ is the simplicial complex whose simplices are the subsets of vertices of $G$ which form a complete subgraph in $G$.  

We denote the set of $k$-simplices in an arbitrary simplicial complex $\Gamma$ by $S_k$. The \textit{up-degree} of a $k$-simplex $\sigma$ is the number of $(k+1)$-simplices which contain $\sigma$, and is denoted by $d_{\textnormal{up}}(\sigma)$. Furthermore, we define the $k^{th}$ \textit{boundary map} as the mapping $\partial_k : \C S_k \rightarrow \C S_{k-1}$ sending a $k$-simplex $\ket{\{x_{i_0},...,x_{i_k}\}}$ with $i_0 < i_1 < ... < i_k$ to
\begin{equation*}
    \partial_k \ket{\{x_{i_0},...,x_{i_k}\}} = \sum_{j=0}^k (-1)^{j+1} \ket{\{x_{i_0},...,x_{i_k}\} \setminus \{x_{i_j}\}},
\end{equation*}
and extended by linearity to $\C S_k$. A basic result in algebraic topology shows that $\partial_k \circ \partial_{k+1} = 0$ \cite[Lemma 2.1]{Hatcher2002}, meaning that $\textnormal{Im}(\partial_{k+1}) \subseteq \ker(\partial_k)$, and therefore the quotient $\ker(\partial_k ) / \textnormal{Im}(\partial_{k+1})$ is well-defined, and called the $k^{th}$ \textit{homology group}. The $k^{th}$ \textit{Betti number} $\beta_k$ is then defined to be the dimension of this as a vector space, that is,
\begin{equation*}
    \beta_k = \dim\Big( \ker(\partial_k ) / \textnormal{Im}(\partial_{k+1})\Big),
\end{equation*}
and is a quantitative expression for the number of $k$-dimensional holes in $\Gamma$. Similarly, the \textit{normalized Betti number} is defined as $\beta_k / |S_k|$. For the remainder of the paper we will only be concerned with estimating the normalized Betti number by defining an estimator $\hat{\beta}_k$ of $\beta_k / |S_k|$. The quality of this estimator will be characterized as follows. 
\begin{definition}
    We say that a random variable $X$ is an $(\epsilon, \eta)$-\textit{estimate} of $\mu \in \R$ if it satisfies
\begin{equation}
\label{eq:betti_estimate}
    \Pr\Big[ \big| X - \mu \big| > \epsilon \Big] < \eta.
\end{equation}
\end{definition}    
For each of the algorithms we consider we will require that the estimator $\hat{\beta}_k$ is an $(\epsilon, \eta)$-estimate. The $k^{th}$ \textit{combinatorial Laplacian} is defined as the mapping $\mathcal{L}_k : \C S_k \rightarrow \C S_k$
\begin{equation}\label{eqn:kLaplacian_boundary}
    \mathcal{L}_k = \partial_k^\dagger \partial_k + \partial_{k+1} \partial_{k+1}^\dagger,
\end{equation}
and we denote the matrix of $\mathcal{L}_k$ with respect to the basis of characteristic vectors of $S_k$ by $\Delta_k$. As a consequence of the Hodge theorem \cite{Eckmann1944} the $k^{th}$ Betti number is equal to the nullity of $\Delta_k$ and therefore can be estimated by analyzing the spectrum of the Laplacian matrix. The entries of $\Delta_k$ have an explicit combinatorial formula \cite[Theorem 3.3.4]{Goldberg2002} which will be important for proving our variance bounds. 

\begin{theorem}
\label{thm:lap_entries}
    Let $\Gamma$ be a finite dimensional oriented simplicial complex. For any $k > 0 $ and $\sigma, \tau \in S_k$ the $(\sigma, \tau)$ entry of $\Delta_k$ is
\begin{equation}
    (\Delta_k)_{\sigma \tau} = \begin{cases}
        d_{\textnormal{up}}(\sigma) + k + 1, & \sigma = \tau \\
        \pm 1, & |\sigma \cap \tau| = k \textnormal{ and } \sigma \cup \tau \notin \Gamma \\
        0, & \textnormal{else}
    \end{cases}
\end{equation}
\end{theorem}

We now sketch the basic ideas around the normalized Betti number estimation algorithm of \cite{Apers2023}. The general idea is to express the normalized Betti number as the trace of a matrix related to the Laplacian and then use a Monte Carlo path integral to estimate its trace. More precisely, the normalized trace of the matrix $H^\ell := (I - \frac{1}{n}\Delta_k)^\ell$ approaches the normalized Betti number as $\ell \rightarrow \infty$, and therefore it suffices to estimate this normalized trace with $\ell$ chosen sufficiently large \footnote{In \cite{Apers2023}, Apers et al. consider the more general matrix $H = I - \frac{1}{\hat{\lambda}}\Delta_k$ where $\hat{\lambda}$ is an upper bound on the largest eigenvalue of $\Delta_k$. In general, the largest eigenvalue of $\Delta_k$ is always at most $n$. We restrict ourselves to the case where $\hat{\lambda}=n$ since the results apply to an arbitrary simplicial complex. Additionally, in this case there exists an analogous quantum algorithm and therefore we may directly compare these algorithms under the same assumptions.} In fact, it was shown in \cite{Apers2023} that it suffices to take $\ell \geq \frac{1}{\delta} \log(1/\epsilon)$ where $\delta$ is the spectral gap of $\Delta_k$ and $\epsilon$ is the desired precision. However, there appears to be little known about the spectral gap of $\Delta_k$ for an arbitrary simplicial complex, hence in practice we must simply choose the path length $\ell$ as large as computationally feasible. Expanding the matrix product $H^\ell$ using the definition of matrix multiplication, this normalized trace can be expressed as the sum
\begin{equation}
    \frac{1}{|S_k|}\tr(H^\ell) = \frac{1}{|S_k|} \sum_{(\sigma_0,...,\sigma_\ell)} \braket{\sigma_\ell|\sigma_0} \prod_{i=0}^{\ell-1} (-1)^{s(\sigma_i, \sigma_{i+1})} \|H\ket{\sigma_i} \|_1 \frac{|\bra{\sigma_{i+1}}H \ket{\sigma_i}|}{\|H\ket{\sigma_i} \|_1},
\end{equation}
where $(-1)^{s(\sigma_i, \sigma_j)}$ with $s(\sigma_i, \sigma_j) \in \{0,1\}$ is the sign of $\bra{\sigma_i}H\ket{\sigma_j}$. The above sum can then be interpreted as an expectation of a particular random variable and then estimated through Monte Carlo sampling. Write $\Omega$ to denote the set $\{ \overline{\sigma}=(\sigma_0,...,\sigma_\ell): \sigma_i \in S_k\}$, which is the set of length $\ell$ paths of $k$-simplices of $\Gamma$. Then define a function $f_k: \Omega \rightarrow \R$ by 
\begin{equation}
    f_k(\overline{\sigma}) = \braket{\sigma_\ell|\sigma_0} \prod_{i=0}^{\ell-1} (-1)^{s(\sigma_i, \sigma_{i+1})} \|H\ket{\sigma_i} \|_1,
\end{equation}
and a probability density $p: \Omega \rightarrow \R$ by
\begin{equation}
\label{eq:path_density}
    p(\overline{\sigma}) = \frac{1}{|S_k|}  \prod_{i=0}^{\ell-1} \frac{|\bra{\sigma_{i+1}}H \ket{\sigma_i}|}{\|H\ket{\sigma_i} \|_1}.
\end{equation}
The above density can be viewed as the probability of sampling a path of $k$-simplices according to a Markov chain with transition matrix $P$, where the  $(\sigma, \tau)$-entry of $P$ represents the probability of transitioning from the $k$-simplex $\sigma$ to $\tau$, and is given by
\begin{equation}
\label{eq:mc_def}
    P_{\sigma \tau} : = \frac{|\bra{\tau}H \ket{\sigma}|}{\|H\ket{\sigma} \|_1},
\end{equation}
and with the initial $k$-simplex chosen uniformly from $S_k$. Under this probabilistic formulation, our normalized trace becomes
\begin{equation}
    \frac{1}{|S_k|}\tr(H^\ell) = \E[f_k(\overline{\bm{\sigma}})],
\end{equation}
where $\overline{\bm{\sigma}}$ is the random variable with density $p$ given in (\ref{eq:path_density}). By expressing the normalized trace as an expectation we obtain an estimator for the normalized Betti number by sampling paths according to the density $p$, applying the function $f$ to the samples, and then averaging the result. This is summarized in Algorithm~\ref{alg:betti}. 

\begin{algorithm}
\caption{Normalized Betti number estimation algorithm \texttt{CBNE} of \cite{Apers2023}}\label{alg:betti}
\begin{algorithmic}[1]
\Let {$\Gamma$ be a simplicial complex with Laplacian $\Delta_k$.} 
\Let {$C=2$ if $\Gamma$ is a clique complex, and otherwise $C=n$.}
\Let {$N_p$ be the number of sample paths.}
\State $N_p \gets \frac{\log(2/\eta)}{\epsilon^2}C^{2\ell}$
\For{$i=1,...,N_p$}
\State Sample a path $\overline{\sigma}_{i}$ of $k$-simplices of length $\ell$ according to the Markov chain $P$. 
\EndFor
\State \Return $\hat{\beta}_k =  \frac{1}{N_p}\sum_{i=1}^{N_p} f_k(\overline{\sigma}_{i})$. 
\Ensure {An estimate $\hat{\beta}_k$ for the normalized Betti number $\beta_k/|S_k|$ of $\Gamma$.}
\end{algorithmic}
\end{algorithm}

To determine the number of samples required to obtain an $(\epsilon, \eta)$-estimate there are several concentration inequalities that can be appealed to. 
\begin{theorem}[Hoeffding Inequality]
\label{thm:hoeffding}
    Suppose that $X_1,...,X_q$ are independent, bounded random variables such that $a \leq X_i \leq b$ almost surely for all $1 \leq i \leq q$, and let $\mu = \E[\frac{1}{q}\sum_{i=1}^q X_i]$. Then for any $\epsilon > 0$ we have
\begin{equation}
    \Pr\Big[ \Big|\frac{1}{q}\sum_{i=1}^q X_i-\mu\Big| \geq \epsilon \Big] \leq 2 \exp\Big(\frac{-2q\epsilon^2}{(b-a)^2} \Big). 
\end{equation}
\end{theorem}

\begin{theorem}[Chebyshev Inequality]
\label{thm:chebyshev}
    Suppose that $X$ is a random variable with expectation $\mu$. Then for any $\epsilon > 0$ we have
\begin{equation}
    \Pr\Big[ \big|X-\mu\big| \geq \epsilon \Big] \leq \frac{\Var[X]}{\epsilon^2}. 
\end{equation}
\end{theorem}

In \cite{Apers2023}, the authors apply the Hoeffding inequality to give an upper bound on the number of sample paths needed for an $(\epsilon, \eta)$-estimate of the normalized Betti number. While applying this inequality has the benefit that it only requires knowledge of the expectation and bounds of the random variable, it is not as sharp as the Chebyshev inequality when the random variable has small variance. On the other hand, the Chebyshev inequality is optimal in the sense that there exists random variables for which the inequality is an equality. To see this more clearly, applying the Hoeffding and Chebyshev inequality to the estimator $f_k(\overline{\bm{\sigma}})$ we obtain
\begin{equation}
    \Pr\Big[ \Big| \frac{1}{N_p}\sum_{i=1}^{N_p} f_k(\overline{\sigma}_{i}) - \E[f_k(\overline{\bm{\sigma}})] \Big| > \epsilon \Big] < 2 \exp\Big(\frac{-2N_p\epsilon^2}{\big(\max(f_k)-\min(f_k)\big)^2} \Big),
\end{equation}
and
\begin{equation}
    \Pr\Big[ \Big| \frac{1}{N_p}\sum_{i=1}^{N_p} f_k(\overline{\sigma}_{i}) - \E[f_k(\overline{\bm{\sigma}})] \Big| > \epsilon \Big] < \frac{\Var[f_k(\overline{\bm{\sigma}})]}{N_p \epsilon^2},
\end{equation}
respectively. Therefore to achieve an $(\epsilon, \eta)$-estimate it suffices to take either $\frac{\ln(2/\eta) \max(f_k)^2}{\epsilon^2}$ or $\frac{\Var[f_k(\overline{\bm{\sigma}})]}{\eta \epsilon^2}$ sample paths. Notably, a sufficient number of samples scales proportionally to either the maximum value squared of $f_k$ or the variance of $f_k(\overline{\bm{\sigma}})$. The challenge here is that it is not obvious if there is a better bound for the variance of $f_k(\overline{\bm{\sigma}})$ other than just the maximum value of the function $f_k$.  

Our goal here will be to determine the asymptotic growth of the variance of $f_k(\overline{\bm{\sigma}})$ as a function of the number of vertices $n$, dimension $k$, and path length $\ell$. Since the variance of a random variable is the second moment minus the square of its expectation, and since the expectation in this case lies in $[0,1]$, then we will primarily focus on the asymptotic growth of the second moment of $f_k(\overline{\bm{\sigma}})$, which will serve as a measure for the number of samples required to estimate the Betti number. This is summarized in the following problem.
\begin{problem}
    Given a finite dimensional oriented simplicial complex $\Gamma$ on $n$ vertices, a dimension $k > 0$, and a path length $\ell$, determine bounds for the second moment $\E[|f_k(\overline{\bm{\sigma}})|^2]$ as a function of $n,k$ and $\ell$. 
\end{problem}

\section{Analyzing the Classical Monte Carlo Algorithm}
\label{sec:CBNE_analysis}

In this section we analyze the Markov chain defined in (\ref{eq:mc_def}) and then give bounds for variance of the estimator used in Algorithm~\ref{alg:betti}. We can recognize this Markov chain as a random walk over a particular graph which we will now define.
\begin{definition}
    The $k$-\textit{simplex graph} of a simpilicial complex $\Gamma$ is the graph denoted by $G(S_k)$ with vertex set $S_k$ with two simplices adjacent if they intersect in a $(k-1)$-simplex but are not contained in a common $(k+1)$-simplex.\footnote{Related graphs over the simplices of a simplicial complex have been defined and studied. For example, in \cite{Incudini2024} a similar signed graph was defined with the same set of vertices although with adjacency defined if the simplices either intersect in a $(k-1)$-simplex or are contained in a common $(k+1)$-simplex. To our knowledge this appears to be the first time that our particular graph over the simplices has been studied.}
\end{definition}
Throughout the rest of the paper if $\sigma$ is a $k$-simplex we will write $\deg(\sigma)$ to denote the degree of $\sigma$ as a vertex in the simplex graph. Note that by Theorem~\ref{thm:lap_entries}, the definition of the simplex graph captures when the off-diagonal entries of $\Delta_k$ are nonzero. We begin with some basic facts about the matrix $H = I - \frac{1}{n}\Delta_k$, the Markov chain $P$, and the simplex graph $G(S_k)$ which will be useful in the later sections.

\begin{theorem}
\label{thm:mc_properties}
Let $\Gamma$ be a finite dimensional oriented simplicial complex on $n$ vertices with simplex graph $G(S_k)$, and let $G(S_k)$ have minimum and maximum degrees $\delta(S_k), \Delta(S_k)$, respectively. The following hold:
\begin{enumerate}[label=(\alph*)]
    \item For any $k$-simplex $\sigma \in S_k$, the 1-norm of the $\sigma$ column of the matrix $H = I - \frac{1}{n}\Delta_k$ is
\begin{equation*}
    \|H\ket{\sigma}\|_1 = 1 + \frac{1}{n}\big(\deg(\sigma) - d_{\textnormal{up}}(\sigma) - k - 1\big).
\end{equation*}
    \item The transition probabilities of the Markov chain $P$ defined in (\ref{eq:mc_def}) are
\begin{equation*}
    P_{\sigma\tau}  = \begin{cases}
    \frac{1}{n + \deg(\sigma) - d_{up}(\sigma) - k -1 } & \sigma, \tau \textnormal{ are adjacent in } G(S_k) \\
    \frac{n - d_{up}(\sigma) - k}{n + \deg(\sigma)- d_{up}(\sigma) - k -1} & \sigma = \tau \\
    0 & \textnormal{else}. \\ 
\end{cases} 
\end{equation*}
    \item The maximum degree of a vertex in $G(S_k)$ is at most $(n-k-1)(k+1)$. If $\Gamma$ is a clique complex of a graph then the maximum degree is at most $n-k-1$. 
\end{enumerate}
\end{theorem}

\begin{proof}
(a) The $1$-norm of the $\sigma$ column of $H$ can be computed from Theorem~\ref{thm:lap_entries} as
\begin{equation*}
    \begin{split}
        \|H\ket{\sigma}\|_1 &= \sum_{\tau}\Big|I-\frac{1}{n}\Delta_k\Big|_{\tau \sigma} \\
        &= \Big|I-\frac{1}{n}\Delta_k\Big|_{\sigma \sigma} + \sum_{\tau \neq \sigma} \Big|\frac{1}{n}\Delta_k \Big|_{\tau \sigma} \\
        &= |1-\frac{1}{n}(d_{\textnormal{up}}(\sigma) + k + 1)| + \frac{1}{n}\deg(\sigma)  \\
        &= 1+\frac{1}{n}(\deg(\sigma)- d_{\textnormal{up}}(\sigma) - k - 1)  \\
    \end{split}
\end{equation*}
(b) From the definition of the Markov chain in (\ref{eq:mc_def}), along with Theorem~\ref{thm:lap_entries}, we obtain
\begin{equation*}
P_{\sigma \tau} : = \frac{|\bra{\tau}H \ket{\sigma}|}{\|H\ket{\sigma} \|_1} = \begin{cases}
    \frac{\frac{1}{n}}{1 + \frac{1}{n}( \deg(\sigma) - d_{up}(\sigma) - k -1)} & \sigma, \tau \textnormal{ are adjacent in } G(S_k) \\
    \frac{1-\frac{1}{n}(d_{\textnormal{up}}(\sigma) + k + 1)}{1 + \frac{1}{n}(\deg(\sigma) - d_{\textnormal{up}}(\sigma)-k -1)}, & \tau = \sigma \\
    0 & \textnormal{else}. \\ 
\end{cases}    
\end{equation*} 
(c) Consider a $k$-simplex $\sigma$ of $\Gamma$. If $\tau$ is another $k$-simplex of $\Gamma$ which forms a neighbor in the simplex graph $G(S_k)$, then $\tau$ can be expressed as $(\sigma\setminus\{x_i\})\cup\{x_j\}$ for some vertices $x_i, x_j$ in $\Gamma$ with $x_i \in \sigma$ and $x_j \notin \sigma$. As each possible neighbor of $\sigma$ leads to a distinct pair $(x_i, x_j)$, and since there are $(n-k-1)(k+1)$ such pairs, then the maximum degree of a vertex in $G(S_k)$ is at most $(n-k-1)(k+1)$. For the case when $\Gamma$ is a clique complex, see \cite[Claim 4.2]{Apers2023}. 
\end{proof}

We now present some general bounds on the second moment of the estimator in Algorithm~\ref{alg:betti}. 

\begin{theorem}
\label{thm:variance_bd_main}
    Let $\Gamma$ be a finite dimensional oriented simplicial complex on $n$ vertices and let $d_{\textnormal{up}}$ denote the maximum up-degree of a $k$-simplex in $\Gamma$. For every dimension $k > 0$, the second moment of $f_k(\overline{\bm{\sigma}})$ is bounded as
\begin{equation*}
    \Big(1 - \frac{d_{\textnormal{up}}+k+1}{n} \Big)^{\ell} \frac{1}{|S_k|} \sum_{\sigma \in S_k} \|H\ket{\sigma}\|_1^{\ell} \leq \E[|f_k(\overline{\bm{\sigma}})|^2] \leq  \frac{1}{|S_k|} \sum_{\sigma \in S_k} \|H\ket{\sigma}\|_1^{2\ell}.
\end{equation*}
Additionally, when applying Algorithm~\ref{alg:betti} to obtain an $(\epsilon, \eta)$-estimate of the normalized trace of $H^\ell$, then it suffices to use
\begin{equation}
    N_p = \frac{1}{\eta\epsilon^2} \times \frac{1}{|S_k|} \sum_{\sigma \in S_k} \|H\ket{\sigma}\|_1^{2\ell}
\end{equation}
sample paths.
\end{theorem}

\begin{proof}
    By definition, the second moment is given by $\E[|f_k(\overline{\bm{\sigma}})|^2] = \sum_{\overline{\sigma} \in \Omega} f(\overline{\sigma})^2p(\overline{\sigma})$. By restricting this sum to the collection of paths which remain at the same $k$-simplex for its entirety, the second moment is at least 
\begin{equation*}
\begin{split}
    \E[|f_k(\overline{\bm{\sigma}})|^2] &\geq \frac{1}{|S_k|} \sum_{\sigma \in S_k} \big\|H\ket{\sigma} \big\|_1^{2\ell} \ \big(P_{\sigma\sigma}\big)^\ell \\
    &=  \frac{1}{|S_k|} \sum_{\sigma \in S_k} \big\|H\ket{\sigma} \big\|_1^\ell \Big(1 - \frac{d_{\textnormal{up}}(\sigma) + k + 1}{n} \Big)^\ell  \\
    &\geq \Big(1 - \frac{d_{\textnormal{up}} + k + 1}{n} \Big)^\ell  \frac{1}{|S_k|} \sum_{\sigma \in S_k} \big\|H\ket{\sigma} \big\|_1^\ell   \\
\end{split}
\end{equation*}
which gives the claimed lower bound. For the upper bound, we apply the inequality of arithmetic and geometric means, which states that $\sqrt[n]{x_1x_2\dots x_n} \leq \frac{1}{n}(x_1+\dots + x_n)$ for any non-negative real numbers $x_1,...,x_n$. Applied to the function $|f_k(\overline{\sigma})|^2$ followed by Jensen's inequality shows that 
\begin{equation*}
    \prod_{i=0}^{\ell-1} \big\|H\ket{\sigma_i} \big\|_1^2 \leq \Big(\frac{1}{\ell} \sum_{i=0}^{\ell-1} \big\|H\ket{\sigma_i} \big\|_1^2 \Big)^{\ell} \leq \frac{1}{\ell} \sum_{i=0}^{\ell-1} \big\|H\ket{\sigma_i} \big\|_1^{2\ell} 
\end{equation*}
for any path $\overline{\sigma}=(\sigma_0,...,\sigma_\ell)\in \Omega$. Additionally, write $\Omega^\prime \subset \Omega$ to denote the set of paths in $\Omega$ with the same starting and ending simplex, and note that $f_k$ is zero outside of $\Omega^\prime$. Therefore we obtain
\begin{equation*}
    \E[|f_k(\overline{\bm{\sigma}})|^2] \leq  \frac{1}{\ell} \sum_{\overline{\sigma} \in \Omega^\prime} \sum_{i=0}^{\ell-1}   \big\|H\ket{\sigma_i} \big\|_1^{2\ell} p(\overline{\sigma}).
\end{equation*}
Exchanging the order of the summations shows
\begin{equation*}
    \E[|f_k(\overline{\bm{\sigma}})|^2] \leq \frac{1}{\ell} \sum_{i=0}^{\ell-1} \sum_{\sigma \in S_k} \big\|H\ket{\sigma} \big\|_1^{2\ell}  \sum_{\overline{\sigma} \in \Omega^\prime, \sigma_i =\sigma} p(\overline{\sigma}).
\end{equation*}
We now claim that $\sum_{\overline{\sigma} \in \Omega^\prime, \sigma_i =\sigma} p(\overline{\sigma}) \leq \frac{1}{|S_k|}$. To prove this, note that this sum represents the probability that a path has the same starting and ending point and meets the simplex $\sigma$ at the $i^{th}$ step. This is equal to $\frac{1}{|S_k|}\sum_{\tau \in S_k}(P^{\ell-i})_{\tau \sigma} (P^i)_{\sigma \tau}$, which applying the definition of matrix multiplication is equal to $\frac{1}{|S_k|}(P^\ell)_{\sigma \sigma}$. Since the entries of $P^\ell$ are probabilities then this quantity is at most $\frac{1}{|S_k|}$. In total, we have shown 
\begin{equation*}
    \E[|f_k(\overline{\bm{\sigma}})|^2] \leq \frac{1}{|S_k|} \sum_{\sigma \in S_k} \big\|H\ket{\sigma} \big\|_1^{2\ell} .
\end{equation*}
as was claimed. Lastly, by the Chebyshev inequality in Theorem~\ref{thm:chebyshev} the number of sample paths sufficient for an $(\epsilon, \eta)$-estimate is
\begin{equation*}
    N_p = \frac{\Var[f_k(\overline{\bm{\sigma}})]}{\eta \epsilon^2} \leq \frac{1}{\eta \epsilon^2} \times  \E[|f_k(\overline{\bm{\sigma}})|^2] \leq \frac{1}{\eta \epsilon^2} \times \frac{1}{|S_k|} \sum_{\sigma \in S_k} \big\|H\ket{\sigma} \big\|_1^{2\ell},
\end{equation*}
completing the proof.
\end{proof}

To illustrate how varied the second moment of $f_k(\overline{\bm{\sigma}})$ can be, we apply the bounds in Theorem~\ref{thm:variance_bd_main} to two simple examples. 

\begin{example}[Complete $(k+1)$-partite graph]
\label{ex:complete_k_partite}
We define the complete $(k+1)$-partite graph to have $(k+1)$ parts, $m$ vertices in each part, and any two vertices adjacent if they lie in different parts. The clique complex of a $(k+1)$-partite graph has dimension $k$. Each $(k+1)$-clique meets $(m-1)(k+1)$ other $(k+1)$-cliques in all but one vertex. Thus
\begin{equation*}
\|H\ket{\sigma}\|_1 = 1 + \frac{(k+1)(m-1)-k-1}{n} =2\Big(1-\frac{k+1}{n}\Big)    
\end{equation*}
for every $k$-simplex $\sigma$. Applying the bounds in Theorem~\ref{thm:variance_bd_main} we obtain
\begin{equation*}
2^\ell\Big(1 - \frac{k+1}{n} \Big)^{2\ell} \leq \E[|f_k(\overline{\bm{\sigma}})|^2] \leq  4^\ell \Big(1 - \frac{k+1}{n} \Big)^{2\ell}.
\end{equation*}
This example shows that the second moment of $f_k(\overline{\bm{\sigma}})$ can be as large as $\Omega(2^\ell)$. 
\end{example}

\begin{example}[Disjoint union of $(k+1)$-cliques]
\label{ex:clique_union}
Consider a graph on $n$ vertices that is a disjoint union of $m$ $(k+1)$-cliques. In this case every $k$-simplex in the clique complex has zero up-degree, and the simplex graph is an edgeless graph with $m$ vertices. Therefore the bounds in Theorem~\ref{thm:variance_bd_main} give
\begin{equation*}
    \E[|f_k(\overline{\bm{\sigma}})|^2] = \Big(1 - \frac{k+1}{n} \Big)^{2\ell}.
\end{equation*}
We can see from this example that the second moment of $f_k(\overline{\bm{\sigma}})$ is less than $1$, which reflects the fact that the Betti numbers in this case are all zero. 
\end{example}

We will now show that the sample complexity derived in Theorem~\ref{thm:variance_bd_main} implies the ones derived by Apers et al. in \cite{Apers2023}. This shows that our sample complexity is no worse than the ones derived there. However, as we will see in Section~\ref{sec:improved_alg}, the upper bounds derived in Theorem~\ref{thm:variance_bd_main} will lead us to an improvement to Algorithm~\ref{alg:betti}. 

\begin{corollary}
    Let $\Gamma$ be a finite dimensional oriented simplicial complex on $n$ vertices. For every dimension $k > 0$, the second moment of $f_k(\overline{\bm{\sigma}})$ is bounded above as
\begin{equation*}
    \E[|f_k(\overline{\bm{\sigma}})|^2] \leq  (k+2)^{2\ell}
\end{equation*}
If $\Gamma$ is a clique complex of a graph on $n$ vertices then for every dimension $k > 0$, the second moment of $f_k(\overline{\bm{\sigma}})$ is bounded above as
\begin{equation*}
    \E[|f_k(\overline{\bm{\sigma}})|^2] \leq  2^{2\ell}\big(1 - \frac{k+1}{n}\big)^{2\ell}.
\end{equation*}
\end{corollary}

\begin{proof}
    For the case of a general simplicial complex, we know from Theorem~\ref{thm:mc_properties}(c) that the maximum degree of a vertex in $G(S_k)$ is at most $(n-k-1)(k+1)$. Applying Theorem~\ref{thm:mc_properties}(a) and ~\ref{thm:variance_bd_main} shows that
\begin{equation*}
    \E[|f_k(\overline{\bm{\sigma}})|^2] \leq \Big(1 + \frac{1}{n}\big((n-k-1)(k+1) - k - 1\big)\Big)^{2\ell} \leq (k+2)^{2\ell}.
\end{equation*}
For the case of a clique complex, the maximum degree of a vertex in $G(S_k)$ is at most $n-k-1$, and we similarly obtain 
\begin{equation*}
    \E[|f_k(\overline{\bm{\sigma}})|^2] \leq \Big(1 + \frac{1}{n}\big((n-k-1) - k - 1\big)\Big)^{2\ell} = 2^{2\ell}\Big(1 - \frac{k+1}{n}\Big)^{2\ell} .
\end{equation*}
\end{proof}

To better understand the asymptotics of the bounds in Theorem~\ref{thm:variance_bd_main} we give sufficient conditions for the upper and lower bounds to be exponential in $\ell$, and present the simplified asymptotics. 

\begin{corollary}
\label{cor:var_bounds_asymp}
    Let $\Gamma$ be a finite dimensional oriented simplicial complex on $n$ vertices with simplex graph $G(S_k)$, and let $G(S_k)$ have minimum and maximum degrees $\delta(S_k), \Delta(S_k)$, respectively. Suppose that $\Gamma$ has dimension $k$, and that $k \in o(n)$. Then the variance of $f_k(\overline{\bm{\sigma}})$ is bounded asymptotically as
\[ \Var[f_k(\overline{\bm{\sigma}})] \in \mathcal{O}\Big( \Big(1 + \frac{\Delta(S_k) }{n}\Big)^{2\ell} \Big),  \]
\[ \Var[f_k(\overline{\bm{\sigma}})] \in \Omega\Big(\Big(1 + \frac{\delta(S_k)}{n}\Big)^\ell \Big)\]
\end{corollary}

\begin{proof}
    Since $\Gamma$ is assumed to have dimension $k$, then every $k$-simplex has up-degree zero. Therefore the second moment bounds of Theorem~\ref{thm:variance_bd_main} implies the bounds 
\begin{equation*}
    \Big(1 - \frac{k+1}{n} \Big)^{\ell} \Big(1 + \frac{\delta(S_k)- k -1}{n} \Big)^{\ell} \leq \E[|f_k(\overline{\bm{\sigma}})|^2] \leq  \Big(1 + \frac{\Delta(S_k)- k -1}{n} \Big)^{2\ell}
\end{equation*}
Since $k \in o(n)$ then this gives the claimed asymptotic bounds.
\end{proof}

Corollary~\ref{cor:var_bounds_asymp} gives simple conditions for the variance to be exponential in the path length $\ell$. Namely, we require that the simplicial complex has sublinear dimension, and that the minimum degree of the simplex graph is in $\Omega(n)$.

\section{An improved Betti number estimation algorithm}
\label{sec:improved_alg}

 As noted in Section~\ref{sec:background}, to obtain an $(\epsilon,\eta)$-estimate of the expectation of $f_k(\overline{\bm{\sigma}})$ we only require a number of sample paths proportional to the variance of $f_k(\overline{\bm{\sigma}})$. However, the sample complexity of the \texttt{CBNE} algorithm described in \cite{Apers2023} is computed using the Hoeffding inequality, which is a weaker inequality for random variables with small variance. For a general simplicial complex the number of samples derived there for an $(\epsilon, \eta)$-estimate of the expectation is $\mathcal{O}(\frac{\log(2/\eta)}{\epsilon^2}n^{2\ell})$ \cite[Lemma 3.4]{Apers2023}, and for a clique complex the sample complexity is $\mathcal{O}(\frac{\log(2/\eta)}{\epsilon^2}4^{\ell})$ \cite[Section 4]{Apers2023}. We can see from the Examples~\ref{ex:complete_k_partite} and \ref{ex:clique_union} presented in the previous section that these general bounds can be significantly larger than the variance and thus significantly overestimate the number of samples needed. One of the main reasons for this shortcoming is that, previously, there were no known reasonable estimates for the variance.

In this section we propose an improvement to Algorithm~\ref{alg:betti} that we refer to as \texttt{CBNE-Var}, which remedies this and reduces the sample complexity in cases where the variance is small. The idea behind our algorithm is to first use a number of simplex samples to estimate the variance upper bound presented in Theorem~\ref{thm:variance_bd_main}, and then only sample that number of paths for the final Betti number estimate. 

The upper bound in Theorem~\ref{thm:variance_bd_main} can be interpreted as an expectation and therefore naturally leads to a Monte Carlo estimator. That is, if $\bm{\sigma}$ is a $k$-simplex chosen uniformly from $S_k$, then the expectation $\E[\|H\ket{\bm{\sigma}}\|_1^{2\ell}]$ is precisely the variance upper bound $\frac{1}{|S_k|}\sum_{\sigma \in S_k} \|H\ket{\sigma}\|_1^{2\ell}$ described in Theorem~\ref{thm:variance_bd_main}. Therefore, by repeatedly sampling a $k$-simplex uniformly from $S_k$, computing the function $\|H\ket{\sigma}\|_1^{2\ell}$, and averaging the result we obtain an estimate for the variance upper bound. We now determine a sufficient number of samples to estimate this quantity. 

\begin{lemma}
\label{lem:betti2_num_samples}
    Let $\epsilon, \eta > 0$, and for $1 \leq i \leq N_s$ let $\bm{\sigma}_i$ be $k$-simplices chosen i.i.d. uniformly from $S_k$. Additionally, let $C$ be an upper bound for $\|H\|_1$. If $N_s \geq \big(\frac{C^{2\ell}}{\epsilon}\big)^2 \ln(1/\eta)$, then the random variable 
\begin{equation*}
    \hat{V}:=\frac{1}{N_s}\sum_{i=1}^{N_s} \|H\ket{\bm{\sigma}_i}\|_1^{2\ell}
\end{equation*}
is an $(\epsilon, \eta)$-estimator for 
\begin{equation*}
    \frac{1}{|S_k|}\sum_{\sigma \in S_k} \|H\ket{\sigma}\|_1^{2\ell}.
\end{equation*}
\end{lemma}

\begin{proof}
    By hypothesis, the random variable $\|H\ket{\bm{\sigma}_i}\|_1^{2\ell}$ is bounded as $0 \leq \|H\ket{\bm{\sigma}_i}\|_1^{2\ell} \leq C^{2\ell}$. We apply Hoeffding's inequality in Theorem~\ref{thm:hoeffding}, which implies
\begin{equation*}
    \begin{split}
        \Pr \Big[ \Big| \frac{1}{N_s}\sum_{i=1}^{N_s} \|H\ket{\bm{\sigma_i}}\|_1^{2\ell} - \E[\|H\ket{\bm{\sigma}}\|_1^{2\ell}] \Big| > \epsilon \Big] < 2 \exp\Big(-\frac{2N_s \epsilon^2}{C^{4\ell}} \Big)
    \end{split}
\end{equation*}
Therefore to make the right side at most $\eta$ it suffices to take $N_s$ as claimed. 
\end{proof}

With this result established, we provide some intuition behind the number of samples and sample paths used in Algorithm~\ref{alg:betti2} before proving its correctness and sample complexity. To determine the optimal number of simplex samples to bound the variance, suppose that we first sample $N_s$ $k$-simplices for an estimate $\hat{V}$ of $\E[\|H\ket{\bm{\sigma}}\|_1^{2\ell}]$, which is precisely the variance upper bound described in Theorem~\ref{thm:variance_bd_main}. From Lemma~\ref{lem:betti2_num_samples} we know that $\hat{V}$ estimates its expectation with error $C^{2\ell}/\sqrt{N_s}$. Since this estimate is potentially smaller than its expectation, we correct for this by adding $C^{2\ell}/\sqrt{N_s}$ to the estimate $\hat{V}$. Therefore, in estimating the normalized trace of $H^\ell$ the total sample complexity is
\begin{equation*}
   N_s + \frac{1}{\eta\epsilon^2}\Big(\hat{V} + \frac{C^{2\ell}}{\sqrt{N_s}}\Big) \leq \frac{1}{\eta\epsilon^2} \times \Big( \frac{1}{|S_k|} \sum_{\sigma \in S_k} \|H\ket{\sigma}\|_1^{2\ell} \Big) + N_s + \frac{2C^{2\ell}}{\eta\epsilon^2\sqrt{N_s}}.
\end{equation*}
The number of samples $N_s$ described in Algorithm~\ref{alg:betti2} is chosen to minimize the term $N_s + \frac{2C^{2\ell}}{\eta\epsilon^2\sqrt{N_s}}$, which is minimized when $N_s = \frac{C^{4\ell/3}}{\eta^{2/3}\epsilon^{4/3}}$.  

\begin{theorem}
    Executing Algorithm~\ref{alg:betti2} produces an $(\epsilon, \eta)$-estimate for the normalized trace of $H^\ell$. The total sample complexity is
\begin{equation}
    \frac{3n^{4\ell/3}}{\eta^{2/3}\epsilon^{4/3}} + \frac{1}{\eta\epsilon^2} \times \frac{1}{|S_k|} \sum_{\sigma \in S_k} \|H\ket{\sigma}\|_1^{2\ell}
\end{equation}
when applied to a general simplicial complex, and
\begin{equation}
    \frac{3 \cdot 2^{4\ell/3}}{\eta^{2/3}\epsilon^{4/3}} + \frac{1}{\eta\epsilon^2} \times \frac{1}{|S_k|} \sum_{\sigma \in S_k} \|H\ket{\sigma}\|_1^{2\ell}
\end{equation}
for a clique complex.
\end{theorem}

\begin{proof}
    We first show that Algorithm~\ref{alg:betti2} produces an $(\epsilon, \eta)$-estimate for the normalized trace of $H^\ell$. According to Theorem~\ref{thm:variance_bd_main}, a sufficient number of sample paths to estimate the normalized trace of $H^\ell$ is 
\begin{equation*}
    \frac{1}{\eta \epsilon^2} \times \frac{1}{|S_k|} \sum_{\sigma \in S_k} \|H\ket{\sigma}\|_1^{2\ell}.
\end{equation*}
Therefore it suffices to check that the number of sample paths used in Algorithm~\ref{alg:betti2} is at least this quantity. The number of sample paths used there is
\begin{equation*}
    N_p = \frac{1}{\eta\epsilon^2}\Big(\hat{V} + \frac{C^{2\ell}}{\sqrt{N_s}}\Big)
\end{equation*}
By Lemma~\ref{lem:betti2_num_samples}, $\hat{V}$ estimates its expectation with error $\frac{C^{2\ell}}{\sqrt{N_s}}$, hence the quantity  $\hat{V} + \frac{C^{2\ell}}{\sqrt{N_s}}$ is at least its expectation. We therefore obtain
\begin{equation*}
N_p \geq \frac{1}{\eta\epsilon^2} \E\big[\|H\ket{\bm{\sigma}}\|_1^{2\ell}\big] = \frac{1}{\eta\epsilon^2} \times \frac{1}{|S_k|} \sum_{\sigma \in S_k} \|H\ket{\sigma}\|_1^{2\ell}.
\end{equation*}
This shows that Algorithm~\ref{alg:betti2} indeed produces an $(\epsilon, \eta)$-estimate for the normalized trace of $H^\ell$. To compute the sample complexity, the total number of samples used is
\begin{equation*}
    N_s + \frac{1}{\eta\epsilon^2}\Big(\hat{V} + \frac{C^{2\ell}}{\sqrt{N_s}}\Big) \leq \frac{1}{\eta\epsilon^2} \times \Big( \frac{1}{|S_k|} \sum_{\sigma \in S_k} \|H\ket{\sigma}\|_1^{2\ell} \Big) + N_s + \frac{2C^{2\ell}}{\eta\epsilon^2\sqrt{N_s}}
\end{equation*}
Substituting the choice of $N_s = C^{\frac{4\ell}{3}}\epsilon^{-\frac{4}{3}} \eta^{-\frac{2}{3}}$ above gives a sample complexity of 
\begin{equation*}
    \frac{3C^{4\ell/3}}{\eta^{2/3}\epsilon^{4/3}} + \frac{1}{\eta\epsilon^2} \times \Big( \frac{1}{|S_k|} \sum_{\sigma \in S_k} \|H\ket{\sigma}\|_1^{2\ell} \Big).
\end{equation*}
Lastly, since $C = 2$ for clique complexes and $C=n$ for a general simplicial complex we obtain the claimed sample complexity. 
\end{proof}

\begin{algorithm}
\caption{\texttt{CBNE-Var}}
\label{alg:betti2}
\begin{algorithmic}[1]
\Let {$\Gamma$ be a simplicial complex with $k^{th}$-order Laplacian $\Delta_k$.} 
\Let {$C=2$ if $\Gamma$ is a clique complex, and otherwise $C=n$.}
\State $N_s \gets C^{\frac{4\ell}{3}}\epsilon^{-\frac{4}{3}} \eta^{-\frac{2}{3}}$
\For{$i=1,..., N_s$} 
\State Sample a $k$-simplex $\sigma_i$ uniformly from $S_k$ and compute $\|H\ket{\sigma_i}\|_1^{2\ell}$.
\EndFor
\State $\hat{V} \gets \frac{1}{N_s}\sum_{i=1}^{N_s} \|H\ket{\sigma_i}\|_1^{2\ell}$ 
\State $N_p \gets \frac{1}{\eta\epsilon^2}(\hat{V} + \frac{C^{2\ell}}{\sqrt{N_s}}) $
\For{$i = 1,..., N_p$} 
\State Sample a path of length $\ell$ of $k$-simplices $\overline{\sigma}_i$ according to $P$.
\EndFor
\Return $\hat{\mu} = \frac{1}{N_p}\sum_{i=1}^{N_p} f_k(\overline{\sigma}_i)$. 
\Ensure {An estimate $\hat{\mu}$ for the normalized trace of $H^\ell$.}
\end{algorithmic}
\end{algorithm}

\section{Erd\H{o}s-Renyi random graphs}
\label{sec:er_graphs}

As described in the previous section, the sample complexity of \texttt{CBNE-Var} depends on the $1$-norm of the columns of $H$, and for graphs where this quantity is small \texttt{CBNE-Var} offers a smaller sample complexity than \texttt{CBNE}. A natural question to ask is how big or small this quantity can be, and whether there are graphs where \texttt{CBNE-Var} achieves this advantage. Similarly, we may also ask if there are graphs where the variance of $f_k(\overline{\sigma})$ is exponential in $\ell$, as this would imply that both \texttt{CBNE} and \texttt{CBNE-Var} require an exponential number of samples.

To answer these questions we consider two Erd\H{o}s-Renyi random graph models. For these models we compute the upper and lower bounds in Theorem~\ref{thm:variance_bd_main} for their clique complexes. Given a random graph model $\mathcal{G}_n$ on $n$ vertices, we will say that $\mathcal{G}_n$ has property $\mathcal{P}$ \textit{almost always} if the probability that $\mathcal{G}_n$ has property $\mathcal{P}$ approaches $1$ as $n$ approaches $\infty$.

\subsection{Multi-partite case}

We first show the existence of graphs for which the variance of $f_k(\overline{\bm{\sigma}})$ is exponential in the path length $\ell$ for arbitrarily large dimension $k$. Our strategy will be to consider a random $(k+1)$-partite graph defined similarly to the standard Erd\H{o}s-Renyi random graph model, and show that this random graph almost always has properties which result in an exponentially growing variance. Since the clique complex of a $(k+1)$-partite graph has dimension at most $k$, then reading off Theorem~\ref{thm:variance_bd_main} for a $(k+1)$-partite graph shows that we require a graph whose simplex graph has minimum degree in $\Omega(n)$.
\begin{definition}
\label{def:k_partite_er}
    The $(k+1)$-partite Erd\H{o}s-Renyi random graph, denoted by $G_{n,k,p}$, is the probability space of graphs on vertex set $\{1,...,n\}$ with equal number of vertices in each part, where the probability of an edge occurring between vertices in different parts is $p$, and $0$ otherwise. 
\end{definition}
Given this random graph model we first prove some results regarding the distribution of the degrees of its simplex graph. We will show that the degree of a $(k+1)$-clique is binomially distributed and then show that the degrees of the simplex graph are almost always contained in a small interval.

\begin{lemma}
    The degree of a $(k+1)$-clique $\sigma$ in $G_{n,k,p}$ is distributed as $\textnormal{Binom}(n-k-1, p^k)$. That is, if $\sigma$ is a $(k+1)$-clique then the probability that it has degree $i$ in the simplex graph is 
\begin{equation*}
    \Pr\big[\deg(\sigma) = i| \sigma \textnormal{ is a clique}\big] = \Pr\Big[ \textnormal{Binom}(n-k-1, p^k) = i\Big].
\end{equation*} 
\end{lemma}

\begin{proof}
First note that if $\sigma$ is a $(k+1)$-clique then it must have exactly one vertex from each part in the partition. If, in addition, $\sigma$ has degree $i$, then there is a set of vertices $\tau \subseteq V \setminus \sigma$ of size $i$ where each vertex in $\tau$ is adjacent to all but one vertex of $\sigma$, and the remaining vertices in $ V \setminus (\tau \cup \sigma)$ do not satisfy this property. For $t \in V \setminus \sigma$ we let $A_t$ be the event that the vertex $t$ is adjacent to all but one vertex of $\sigma$. Then our probability can be expressed as 
\begin{equation*}
\begin{split}
    \Pr\big[\deg(\sigma) = i| \sigma \textnormal{ is a clique}\big]&= \Pr\Big[\bigcup_{\substack{\tau \subseteq V \setminus \sigma \\ |\tau|=i}} \big(\bigwedge_{t \in \tau} A_t \big) \wedge \big( \bigwedge_{t^\prime \in V \setminus (\sigma \cup \tau)} \overline{A_t}\big) \Big] \\
    &= \sum_{\substack{\tau \subseteq V \setminus \sigma \\ |\tau|=i}} \Pr \Big[ \bigwedge_{t \in \tau} A_t \Big] \Pr\Big[ \bigwedge_{t^\prime \in V \setminus (\sigma \cup \tau)} \overline{A_t}\Big] \\ 
    &= \binom{n-k-1}{i} \Pr[A_t]^i \Pr\big[\overline{A}_t\big]^{n-k-1-i} \\
    &= \binom{n-k-1}{i} \Pr[A_t]^i \big(1-\Pr[A_t]\big)^{n-k-1-i}
\end{split}
\end{equation*}
where we have used independence. The above is precisely the PDF of a binomial distribution with parameters $n-k-1$ and $\Pr[A_t]$. To compute $\Pr(A_t)$ for a given a vertex $t$, the probability that $t$ is adjacent to all but one vertex of $\sigma$ is $p^k$ since it cannot be adjacent to the vertex of $\sigma$ that is in the same part of the partition, hence $\Pr[A_t]=p^k$ for any such vertex $t$. 
\end{proof}

Having shown that the degrees of the simplex graph of $G_{n,k,p}$ are distributed as a binomial distribution, the expected degree of a clique is therefore $(n-k-1)p^k$. Choosing $p$ such that $p^k$ is constant then gives an expected degree in $\Omega(n)$. However, it is not immediate that a graph sampled from $G_{n,k,p}$ will have all cliques simultaneously having degree close to this expectation, as required to show exponential variance. In the next result we will show that for a specific regime of parameters $n,k,p$ the $k$-simplex graph of $G_{n,k,p}$ almost always has degree sequence concentrated about this expected value.  
 
\begin{lemma}
\label{lem:degree_interval}
    Let $d = (n-k-1)p^k$. If $k \in o(n/\log(n))$ and $p^k \in \Omega(1)$, then $G_{n,k,p}$ almost always has a simplex graph with minimum degree at least $d/2$.
\end{lemma}

\begin{proof}
We let $T_k$ denote the set of $(k+1)$-cliques whose degree is at least $d/2$ and we will show that $\Pr[T_k \neq S_k] \rightarrow 0$ when $n,k,p$ are chosen according to the statement of the theorem. An application of Markov's inequality shows
\begin{equation*}
    \Pr\big[T_k \neq S_k\big] = \Pr\big[|S_k| - |T_k| \geq 1\big] \leq \E\big[|S_k| - |T_k|\big],
\end{equation*}
so it suffices to show that $\E[|S_k| - |T_k|] \rightarrow 0$. Writing $X_\sigma$ to denote the indicator random variable for the event that the set of vertices $\sigma$ forms a $(k+1)$-clique with degree at least $d/2$, then $\E[|T_k|] = \E[\sum_{\sigma}X_\sigma]$. Therefore by linearity of expectation the expectation of $|T_k|$ is
\begin{equation*}
    \E\big[|T_k|\big] = p^{\binom{k+1}{2}} \sum_{\sigma} \Pr \Big[ \deg(\sigma) \geq d/2 \Big].
\end{equation*}
Since the degree distribution of a clique is binomial then the expected degree of a clique $\sigma$ is 
\begin{equation*}
    \E\big[\deg(\sigma)| \sigma \textnormal{ is a clique}\big] = (n-k-1)p^k = d.
\end{equation*}
We apply the Chernoff bound to determine how concentrated the degree of a clique is around its expectation. This shows
\begin{equation*}
    \Pr \Big[ \deg(\sigma) \geq d/2 \Big] \leq \exp\Big(\frac{-d}{8} \Big).    
\end{equation*}
Therefore the expectation of $|T_k|$ is bounded as
\begin{equation*}
    \E\big[|T_k|\big] \geq \Big(\frac{n}{k+1}\Big)^{k+1} p^{\binom{k+1}{2}} \Big(1-\exp\Big(\frac{-d}{8} \Big)\Big)  = \E\big[|S_k|\big] \Big(1-\exp\Big(\frac{-d}{8} \Big)\Big), 
\end{equation*}
which gives the bound
\begin{equation}
\label{eq:multi_part_er_expectation_bd}
    \E\big[|S_k|-|T_k|\big] \leq \E[|S_k|] \exp\Big(\frac{-d}{8} \Big) = \Big(\frac{n}{k+1}\Big)^{k+1} \exp\Big(\frac{-d}{8} \Big).
\end{equation}
Lastly, we compare the asymptotic growth of $\big(\frac{n}{k+1}\big)^{k+1}$ and $\exp(d/8)$. Since $p^k \in \Omega(1)$ then $d \in \Omega(n)$ and therefore $\exp(d/8) \in \Omega(\exp(n))$. The assumption that $k \in o(n/\log(n))$ implies that $\big(\frac{n}{k+1}\big)^{k+1} \in  o(\exp(n))$, which therefore shows that the right side of (\ref{eq:multi_part_er_expectation_bd}) goes to $0$ as $n\rightarrow\infty$, completing the proof. 
\end{proof}

We can now use Lemma~\ref{lem:degree_interval} to show the existence of a $(k+1)$-partite graph where the variance of $f_k(\overline{\bm{\sigma}})$ grows exponentially in the path length.

\begin{theorem}
\label{thm:main}
Suppose that $k \in o(n/\log(n))$ and $p^k \in \Omega(1)$. The $(k+1)$-partite random graph $G_{n,k,p}$ almost always has
\begin{equation*}
 \Var[f_k(\overline{\bm{\sigma}})] \in \Omega\Big(\Big(1 + p^k \Big)^\ell\Big).
\end{equation*}
\end{theorem}

\begin{proof}
Applying Theorem~\ref{thm:variance_bd_main} to $G_{n,k,p}$ shows
\begin{equation*}
\begin{split}
   \E[|f_k(\overline{\bm{\sigma}})|^2] \geq  \Big(1 - \frac{k+1}{n} \Big)^{\ell} \frac{1}{|S_k|} \sum_{\sigma \in S_k} \Big( 1 + \frac{\deg(\sigma)-k-1}{n} \Big)^{\ell},
\end{split}
\end{equation*}
where $d_{\textnormal{up}}(\sigma)=0$ for any $k$-simplex $\sigma$ since $G_{n,k,p}$ is $(k+1)$-partite. From Lemma~\ref{lem:degree_interval} and the assumption that $k \in o(n/\log(n))$ and $p^k \in \Omega(1)$, $G_{n,k,p}$ almost always has a simplex graph with minimum degree at least $\frac{1}{2}(n-k-1)p^k$. Therefore we have
\begin{equation*}
\begin{split}
   \E[|f_k(\overline{\bm{\sigma}})|^2] \geq  \Big(1 - \frac{k+1}{n} \Big)^{\ell} \Big( 1 + \frac{(n-k-1)p^k}{2n} - \frac{k+1}{n} \Big)^{\ell}.
\end{split}
\end{equation*}
From the assumption that $k \in o(n/\log(n))$ and $p^k \in \Omega(1)$, this lower bound is asymptotically 
\begin{equation*}
     \Big(1 - \mathcal{O}(1/n) \Big)^{\ell} \Big( 1 + p^k - \mathcal{O}(1/n) \Big)^\ell,
\end{equation*}
which is equal to the claimed asymptotic expression.
\end{proof}

\begin{example}
    We can apply Theorem~\ref{thm:main} with $k \in o(n/\log(n))$ and $p = 1 - 1/k$. In this case $p^k \geq 1/e$, so that $p^k \in \Omega(1)$. Applying the theorem this shows that there exists a $(k+1)$-partite graph on $n$ vertices for which
\begin{equation*}
 \Var[f_k(\overline{\bm{\sigma}})] \in \Omega\big((1 + 1/e)^\ell \big).
\end{equation*}
\end{example}

\subsection{General case}

In this section we produce graphs for which the \texttt{CBNE-Var} algorithm has an improved sample complexity when compared to \texttt{CBNE}. For this we now consider a general Erd\H{o}s-Renyi graph, defined as follows.

\begin{definition}
\label{def:er}
    The Erd\H{o}s-Renyi random graph, denoted by $G_{n,p}$, is the probability space of graphs on vertex set $\{1,...,n\}$, where every edge occurs independently with probability $p$. 
\end{definition}

Given this random graph model we first prove some results regarding the distribution of the degrees of its simplex graph and the up-degree of a $k$-simplex. Unlike the previous situation of a $(k+1)$-partite graph which easily allowed us to restrict the dimension of the clique complex, in this more general setting we have no such restriction to simplify the computations. Nevertheless, we will show that for a given $k$-simplex $\sigma$ the random variable $\deg(\sigma)-d_{\textnormal{up}}(\sigma)$ can be expressed as a more manageable sum of random variables, and then show that almost surely it is concentrated around its expected value, which will then allow us to upper-bound the variance of $f_k(\overline{\bm{\sigma}})$.

\begin{lemma}
\label{lem:gen_er_deg_distr}
    Suppose that $\sigma$ is a $(k+1)$-clique in $G_{n,p}$. Then the random variable $\deg(\sigma)-d_{\textnormal{up}}(\sigma)$ can be expressed as the sum of random variables 
\begin{equation*}
    \deg(\sigma)-d_{\textnormal{up}}(\sigma) = \sum_{t \notin \{1,2,...,n\}\setminus \sigma} X_t,
\end{equation*}
where the $X_t$ are i.i.d. random variables taking value $1$ with probability $(k+1)p^k(1-p)$; $-1$ with probability $p^{k+1}$; and $0$ with probability $1 - (k+1)p^k(1-p) - p^{k+1}$. 
\end{lemma}

\begin{proof}
    For every vertex $t \in \{1,2,...,n\}\setminus \sigma$ there are three distinct and mutually-exclusive possibilities: That $(\sigma \cup \{t\}) \setminus \{s\}$ is a $k$-simplex for some $s \in \sigma$ and $\sigma \cup \{t\}$ is not a $(k+1)$-simplex; that $\sigma \cup \{t\}$ is a $(k+1)$-simplex; and neither of these two previous cases holding. These correspond to precisely three distinct outcomes regarding the adjacency of $t$: That $t$ is adjacent to all but one vertex in $\sigma$; that $t$ is adjacent to every vertex in $\sigma$; and neither of the first two. 
    
    For $t \in \{1,2,...,n\}\setminus \sigma$, let $Y_t$ be the indicator random variable for the event that $t$ is adjacent to all but one element of $\sigma$. Similarly, define $Z_t$ as the indicator random variable for the event that $t$ is adjacent to every element of $\sigma$. Clearly we have 
\begin{equation*}
    \deg(\sigma)-d_{\textnormal{up}}(\sigma) = \sum_{t \notin \{1,2,...,n\}\setminus \sigma} (Y_t - Z_t).
\end{equation*}
    For distinct vertices $s, t \in \{1,..,n \} \setminus \sigma$ the random variables $Y_s - Z_s$ and $Y_t - Z_t$ are independent since their corresponding events are independent. 
    The probability that $Y_t - Z_t = 1$ is equal to the probability that $t$ is adjacent to all but one vertex of $\sigma$, which is equal to $(k+1)p^k(1-p)$. The probability that $Y_t - Z_t = -1$ is equal to the probability that $t$ is adjacent to every vertex of $\sigma$, which is equal to $p^{k+1}$. Lastly, the probability that $Y_t - Z_t = 0$ is equal to the probability that neither of these two previous events occurs, which gives the probability $1 - (k+1)p^k(1-p) - p^{k+1}$, completing the proof.
\end{proof}

\begin{lemma}
\label{lem:gen_er_deg_interval}
    The expected value of $\deg(\sigma)-d_{\textnormal{up}}(\sigma)$ given that $\sigma$ is a $(k+1)$-clique in $G_{n,p}$ is 
\begin{equation*}
    \mu := (n-k-1)p^k\big( (k+1)-(k+2)p \big).
\end{equation*}
Additionally, if $\epsilon, k$ are such that $k \in o(n/\log(n)^2)$ and $\epsilon = n/\sqrt{\log(n)}$ then $G_{n,p}$ almost always has the property that
\begin{equation*}
    |\deg(\sigma)-d_{\textnormal{up}}(\sigma) - \mu| \leq \epsilon
\end{equation*}
for each $(k+1)$-clique $\sigma$.
\end{lemma}

\begin{proof}
We let $T_k$ denote the set of $(k+1)$-cliques $\sigma$ in $G_{n,p}$ with $|\deg(\sigma)-d_{\textnormal{up}}(\sigma) - \mu| \leq \epsilon$. We will show that $\Pr[T_k \neq S_k] \rightarrow 0$ as $n \rightarrow \infty$ when $k \in o(\epsilon^2/n\log(n))$. An application of Markov's inequality shows
\begin{equation*}
    \Pr\big[T_k \neq S_k\big] = \Pr\big[|S_k| - |T_k| \geq 1\big] \leq \E\big[|S_k| - |T_k|\big],
\end{equation*}
so it suffices to show that $\E[|S_k| - |T_k|] \rightarrow 0$ as $n\rightarrow \infty$. 

For a $(k+1)$-subset $\sigma$ of vertices, we let $X_\sigma$ to denote the indicator random variable for the event that $\sigma$ forms a $(k+1)$-clique and is in $T_k$.  Clearly we have $\E[|T_k|] = \E[\sum_{\sigma}X_\sigma]$. Therefore by linearity of expectation the expectation of $|T_k|$ is
\begin{equation*}
    \E\big[|T_k|\big] = p^{\binom{k+1}{2}} \sum_{\sigma} \Pr \Big[|\deg(\sigma) - d_{\textnormal{up}}(\sigma) - \mu| \leq \epsilon \Big],
\end{equation*}
where the sum is over all $(k+1)$-subsets $\sigma \subseteq \{1,...,n\}$. From Lemma~\ref{lem:gen_er_deg_distr} we know that $\deg(\sigma)-d_{\textnormal{up}}(\sigma)$ is a sum of i.i.d. random variables taking values in $\{-1,0,1\}$ so we can apply the Hoeffding inequality to bound its tail probability. This shows
\begin{equation*}
   \Pr \Big[|\deg(\sigma) - d_{\textnormal{up}}(\sigma) - \mu| \geq \epsilon \Big] \leq 2\exp\Big(-\frac{\epsilon^2}{2(n-k-1)}\Big).
\end{equation*}
Therefore the expectation of $|T_k|$ is bounded as
\begin{equation*}
    \E\big[|T_k|\big] \geq \binom{n}{k+1} p^{\binom{k+1}{2}} \Big(1 - 2\exp\Big(-\frac{\epsilon^2}{2(n-k-1)}\Big)\Big)  = \E\big[|S_k|\big] \Big(1 - 2\exp\Big(-\frac{\epsilon^2}{2(n-k-1)}\Big)\Big), 
\end{equation*}
which gives the bound
\begin{equation*}
    \E\big[|S_k|-|T_k|\big] \leq 2\E[|S_k|]\exp\Big(-\frac{\epsilon^2}{2(n-k-1)}\Big) =  \binom{n}{k+1} p^{\binom{k+1}{2}}2\exp\Big(-\frac{\epsilon^2}{2(n-k-1)}\Big).
\end{equation*}
It now remains to show that this upper bound goes to $0$ as $n \rightarrow \infty$. From the assumption that $k \in o(n/\log(n)^2)$ it follows that $\binom{n}{k+1} \in o(\exp(n/\log(n)))$. The assumption that $\epsilon = n/\sqrt{\log(n)}$ then gives $\exp\Big(-\frac{\epsilon^2}{2(n-k-1)}\Big) \in \mathcal{O}(\exp(-n/\log(n))$. Lastly, since $p \in [0,1]$ then this shows that $\E\big[|S_k|-|T_k|\big] \rightarrow 0$ as $n \rightarrow \infty$, completing the proof. 
\end{proof}

\begin{theorem}
Suppose that $k \in o(n/\log(n)^2)$ and $p \in o(k^{-1/k})$. Then the sample complexity of Algorithm~\ref{alg:betti2} applied to the clique complex of $G_{n,p}$ is almost always $\mathcal{O}\big(\frac{1}{\eta \epsilon^2}2.5^\ell\big)$. 
\end{theorem}

\begin{proof}
We first apply Lemma~\ref{lem:gen_er_deg_interval}. The hypothesis that $k \in o(n/\log(n)^2)$ implies that $G_{n,p}$ almost always has
\begin{equation*}
    \deg(\sigma)-d_{\textnormal{up}}(\sigma) \in [\mu-\epsilon, \mu + \epsilon]
\end{equation*}
for each $(k+1)$-clique $\sigma$, where $\epsilon = n/\sqrt{\log(n)}$ and $\mu$ is given by
\begin{equation*}
    \mu = (n-k-1)p^k((k+1) - (k+2)p).
\end{equation*}
The quantity $\mu$ is at most $n(k+1)p^k$. Applying the upper bound for the second moment given in Theorem~\ref{thm:variance_bd_main} we find that 
\begin{equation*}
    \E[|f_k(\bm{\sigma})|^2] \leq \Big(1 + \frac{\mu + \epsilon - k - 1}{n} \Big)^{2\ell} \leq \Big(1 + \frac{\mu + \epsilon}{n} \Big)^{2\ell}.
\end{equation*}
From the assumption that $p \in o(k^{-1/k})$ it follows that $(k+1)p^k \in o(1)$, which then implies that $\frac{\mu+\epsilon}{n} \in o(1)$. Therefore for $n$ large enough the second moment of $f_k(\bm{\sigma})$ is at most $2.5^\ell$. In this case, the sample complexity of Algorithm~\ref{alg:betti2} is in $\mathcal{O}\big(\frac{1}{\eta \epsilon^2}2.5^\ell\big)$.
\end{proof}

\begin{table}
\centering
\begin{tabular}{ |c|c|c|c|c|c| } 
 \hline
  Graph model                     & Dimension $k$  & Probability $p$       & \texttt{CBNE}                     & \texttt{CBNE-Var}                \\
 \hline
  Erd\H{o}s-Reyni $G_{n,p}$       & $o(n/\log(n)^2)$ & $o(k^{-1/k})$ & $\mathcal{O}\Big(\frac{1}{\eta \epsilon^2} 4^\ell\Big) $   & $\mathcal{O}\Big(\frac{1}{\eta \epsilon^2}2.5^\ell \Big)$ \\ 
  $(k+1)$-partite Erd\H{o}s-Reyni $G_{n,k,p}$ & $o(n/\log(n))$ & $\omega((k/n)^{1/k})$ & $\Omega\Big(\frac{1}{\eta \epsilon^2}(1 + p^k)^\ell \Big)$ & $\Omega\Big(\frac{1}{\eta \epsilon^2}(1 + p^k)^\ell \Big)$\\
 \hline
\end{tabular}
\caption{A summary of the results presented in Section~\ref{sec:er_graphs}. The first column lists the random graph models defined in Definitions~\ref{def:er} and \ref{def:k_partite_er}. The last two columns list the sample complexities for the corresponding algorithms using the dimension $k$ and probablity parameter $p$ in the second and third columns.}
\label{tab:er_summary}
\end{table}

\section{Conclusion}

In this paper we studied the sample complexity of the normalized Betti number estimation presented in Algorithm~\ref{alg:betti}. This is an important problem since this specific algorithm has an analogous quantum algorithm \cite{Akhalwaya2024}, and understanding the precise sample complexity can pave the way for finding regimes for quantum advantage. To this end, we analyzed the variance of the Monte Carlo estimators utilized in this Betti number estimation algorithm and showed that the variance is mainly governed by the amount of intersection of the simplices in the complex. Following this, we then showed that these bounds lead to a modified Betti number algorithm which has a smaller sample complexity for cases where the variance is small. Lastly, we produced random graphs which have large and small variance.

An interesting avenue for further research is to determine to what extent these results might improve the quantum algorithm for normalized Betti number estimation presented in \cite{Akhalwaya2024}. For example, the sample complexities for the quantum algorithm proposed there does not depend on any combinatorial properties of the simplicial complex. It would be interesting to determine if there is a similar variance analysis as the one done here that might lead to a sample complexity reduction in certain cases.

\section*{Acknowledgements}

The author would like to thank Ismail Akhalwaya, Ahmed Bhayat, and Adam Connolly for several discussions which improved the content of this paper. Additionally, we thank Marcello Benedetti and Fred Sauvage for carefully reviewing a draft of this paper and offering many useful suggestions.

\bibliographystyle{IEEEtran}
\bibliography{mybib}{}

% Generated by IEEEtran.bst, version: 1.14 (2015/08/26)
\begin{thebibliography}{10}
\providecommand{\url}[1]{#1}
\csname url@samestyle\endcsname
\providecommand{\newblock}{\relax}
\providecommand{\bibinfo}[2]{#2}
\providecommand{\BIBentrySTDinterwordspacing}{\spaceskip=0pt\relax}
\providecommand{\BIBentryALTinterwordstretchfactor}{4}
\providecommand{\BIBentryALTinterwordspacing}{\spaceskip=\fontdimen2\font plus
\BIBentryALTinterwordstretchfactor\fontdimen3\font minus \fontdimen4\font\relax}
\providecommand{\BIBforeignlanguage}[2]{{%
\expandafter\ifx\csname l@#1\endcsname\relax
\typeout{** WARNING: IEEEtran.bst: No hyphenation pattern has been}%
\typeout{** loaded for the language `#1'. Using the pattern for}%
\typeout{** the default language instead.}%
\else
\language=\csname l@#1\endcsname
\fi
#2}}
\providecommand{\BIBdecl}{\relax}
\BIBdecl

\bibitem{Schmidhuber2023}
A.~Schmidhuber and S.~Lloyd, ``Complexity-theoretic limitations on quantum algorithms for topological data analysis,'' \emph{PRX Quantum}, vol.~4, p. 040349, Dec 2023.

\bibitem{Friedman1998}
J.~Friedman, ``Computing betti numbers via combinatorial laplacians,'' \emph{Algorithmica}, vol.~21, no.~4, pp. 331--346, Aug 1998.

\bibitem{Lloyd2016}
S.~Lloyd, S.~Garnerone, and P.~Zanardi, ``Quantum algorithms for topological and geometric analysis of data,'' \emph{Nature Communications}, vol.~7, 2016.

\bibitem{Berry2024}
D.~W. Berry, Y.~Su, C.~Gyurik, R.~King, J.~Basso, A.~D.~T. Barba, A.~Rajput, N.~Wiebe, V.~Dunjko, and R.~Babbush, ``Analyzing prospects for quantum advantage in topological data analysis,'' \emph{PRX Quantum}, vol.~5, p. 010319, Feb 2024.

\bibitem{Crichigno2022}
M.~Crichigno and T.~Kohler, ``{Clique Homology is QMA1-hard},'' 9 2022.

\bibitem{McArdle2022}
S.~McArdle, A.~Gily{\'e}n, and M.~Berta, ``{A streamlined quantum algorithm for topological data analysis with exponentially fewer qubits},'' 9 2022.

\bibitem{Apers2023}
S.~Apers, S.~Gribling, S.~Sen, and D.~Szab{\'{o}}, ``A (simple) classical algorithm for estimating {B}etti numbers,'' \emph{{Quantum}}, vol.~7, p. 1202, Dec. 2023.

\bibitem{Akhalwaya2024}
\BIBentryALTinterwordspacing
I.~Y. Akhalwaya, A.~Bhayat, A.~Connolly, S.~Herbert, L.~Horesh, J.~Sorci, and S.~Ubaru, ``Comparing quantum and classical monte carlo algorithms for estimating betti numbers of clique complexes,'' 2024. [Online]. Available: \url{https://arxiv.org/abs/2408.16934}
\BIBentrySTDinterwordspacing

\bibitem{Lim2020}
L.-H. Lim, ``Hodge laplacians on graphs,'' \emph{SIAM Review}, vol.~62, no.~3, pp. 685--715, 2020.

\bibitem{Hatcher2002}
A.~Hatcher, \emph{Algebraic topology}.\hskip 1em plus 0.5em minus 0.4em\relax Cambridge: Cambridge University Press, 2002.

\bibitem{Eckmann1944}
B.~Eckmann, ``Harmonische funktionen und randwertaufgaben in einem komplex.'' \emph{Commentarii mathematici Helvetici}, vol.~17, pp. 240--255, 1944/45.

\bibitem{Goldberg2002}
T.~Goldberg, \emph{Combinatorial Laplacians of simplicial complexes}.\hskip 1em plus 0.5em minus 0.4em\relax Bard College, 2002.

\bibitem{Incudini2024}
M.~Incudini, C.~Gyurik, R.~Molteni, and V.~Dunjko, ``{Testing the presence of balanced and bipartite components in a sparse graph is QMA1-hard},'' 12 2024.

\end{thebibliography}

\end{document}